\documentclass{article}
\usepackage{ijcai23}

\usepackage{times}
\usepackage{soul}
\usepackage{url}
\usepackage[hidelinks]{hyperref}
\usepackage[utf8]{inputenc}
\usepackage[small]{caption}
\usepackage{graphicx}
\usepackage{amsmath}
\usepackage{amsthm}
\usepackage{booktabs}
\usepackage{algorithm}
\usepackage{algorithmic}
\usepackage{amsfonts}
\usepackage{amssymb}
\usepackage{framed}
\usepackage{tikz}
\usepackage[switch]{lineno}

\newtheorem{example}{Example}
\newtheorem{theorem}{Theorem}
\newtheorem{definition}{Definition}
\newtheorem{observation}{Observation}

\newtheorem{corollary}{Corollary}

\title{Incentive-Compatible Selection for One or Two Influentials}

\author{
Yuxin Zhao
\and
Yao Zhang\And
Dengji Zhao
\affiliations
ShanghaiTech University
\emails
\{zhaoyx5, zhangyao1, zhaodj\}@shanghaitech.edu.cn
}
\begin{document}

\maketitle

\begin{abstract}
Selecting influentials in networks against strategic manipulations has attracted many researchers' attention and it also has many practical applications. %The existing studies mostly focused on selecting a single influential.
Here, we aim to select one or two influentials in terms of progeny (the influential power) and prevent agents from manipulating their edges (incentive compatibility). The existing studies mostly focused on selecting a single influential for this setting.
\citeauthor{zhang2021incentive}~\shortcite{zhang2021incentive} studied the problem of selecting one agent and proved an upper bound of $1/(1+\ln 2)$ to approximate the optimal selection. In this paper, we first design a mechanism to actually reach the bound. Then, we move this forward to choosing two agents and propose a mechanism to achieve an approximation ratio of $(3+\ln2)/(4(1+\ln2))$ ($\approx 0.54$).

% work on this specific problem leaves a gap between the approximation ratios of existing mechanisms and the proven upper bound for selecting a single agent, and unknown mechanisms for selecting multiple agents. Therefore, we first propose an optimal mechanism that achieves the proven upper bound ($1/(1+\ln 2)$) for selecting a single agent. Then, we also present an extended mechanism for selecting up to two agents, which achieves an approximation ratio of $(3+\ln2)/(4(1+\ln2))$ ($\approx 0.54$).

% Incentive-compatible (IC) selection aims to select a group of agents based on peer nomination, in which participants can not strategically improve their selection probabilities. In this paper, we focus on selecting influentials with larger progeny while maintaining the IC property. Approximation ratio between the expected and the optimal progeny is used to evaluate the quality of IC selection.
% We propose a series of IC selection mechanisms by assigning probabilities to strategic agents. One of these mechanisms is proved to be optimal when selecting at most one influential and has an approximation ratio of $1/(1+\ln2)$. We also present the first IC mechanism that selects at most two influentials with larger progeny. This mechanism achieves an approximation ratio of 0.54.
\end{abstract}

\section{Introduction}
Consider the scenario where we want to select influential agents in a network constructed by referral relationships (e.g., the following relationships in Twitter, the citations between academic papers, etc.). The selected agents may be rewarded with prizes or benefits (e.g., job opportunities~\cite{kotturi2020hirepeer}). Hence, agents have the incentive to manipulate their relationships to make themselves selected. Therefore, selection mechanisms that can prevent agents from strategic manipulations (which is referred to as the property of incentive compatibility) are highly demanded~\cite{alon2011sum}.
% However, because of the double identities, those who desire to win have chances to act strategically, such as hiding their true votes or forming coalitions, in order to outperform their strong rivals. The strategic manipulations of members could affect the trend and result in a poor election outcome. This problem is also widespread in other peer selection applications. To tackle this problem, our goal is to design IC mechanisms that can guarantee all participants truthfully report their types and thus can output the influentials. The word `influential' refers to people with high influence powers, such as a committee chair, a popular user with millions of followers on social media, or an academic paper with numerous quotes. These influentials have great significance. A committee chair influences policy, a popular user boosts advertisement revenue, and a highly cited paper provides academic value. We aim to find such influentials by IC peer selection mechanisms. 

Many studies have investigated incentive-compatible selection mechanisms on different influence measurements for different purposes (see \cite{olckers2022manipulation} for a complete survey). %Most of them look for influentials that have maximal indegree.
In this paper, we focus on the setting where an agent's influential power is measured by her progeny (the number of all agents who directly or indirectly follow her).
%The progeny measurement is more challenging than the indegree one because it considers both direct and indirect influence powers.
For this setting, two studies have been conducted before. \citeauthor{babichenko2020incentive2}~\shortcite{babichenko2020incentive2} proposed the first single agent selection mechanism for progeny maximization that can prevent agents from adding or hiding their out-edges. Their mechanism reaches an approximation ratio of about $1/3$ (i.e., the expected progeny of the chosen agent is about $1/3$ of the largest). However, their mechanism only works in forests. Therefore, \citeauthor{zhang2021incentive}~\shortcite{zhang2021incentive} further studied the same problem in directed acyclic graphs (DAGs) with restricting manipulations in the scope of hiding edges (agents cannot add new edges). 
%It is reasonable for applications where agents cannot follow unrelated agents or those they do not know. 
Their proposed mechanism achieves an approximation ratio of $1/2$. Moreover, they proved an upper bound $1/(1+\ln2)$ of the approximation ratio for any incentive-compatible and fair selection mechanism in the DAG setting. 

In this paper, we follow the DAG setting of \cite{zhang2021incentive} and make the following contributions:
\begin{itemize}
    \item For selecting one agent, we close the gap between the known approximation ratio of $1/2$ and the upper bound of $1/(1+\ln2)$. We propose a mechanism to achieve the exact upper bound.
    \item For selecting two agents, we show that, for the class of mechanisms that only select agents from the $1$-influential set\footnote{$1$-influential set contains all agents each of whom has the largest progeny by hiding her out-edges.} (most of the existing mechanisms belong to this class), the approximation ratio cannot exceed $1/2$ if the target is to select at most two agents. Moreover, we provide a deterministic mechanism in this class that exactly reaches the approximation ratio of $1/2$.
    \item We then propose a new incentive-compatible mechanism based on a $2$-influential set for selecting two agents. The new mechanism achieves a higher approximation ratio of $(3+\ln2)/(4(1+\ln2))$ ($\approx 0.54$). We also provide a general upper bound ($23/27$) of any incentive-compatible mechanism for selecting two agents.
\end{itemize}

\subsection{Other Related Work}
Many studies on incentive-compatible selection mechanisms use in-degrees to measure agents' influential power, which is also referred to as peer selection.
% Incentive compatibility (also called strategyproof or impartiality) was introduced as an economic property in the work of~\cite{de2008impartial} for dollar division.
For the in-degree measurement, \citeauthor{alon2011sum}~\shortcite{alon2011sum} firstly proposed an incentive-compatible peer selection mechanism by a randomized partition method, which divides agents into two groups and chooses the agents according to their in-degrees from the other group. %first applied this property to the peer selection problem.
% Different from our setting, the goal here is to select agents with a good approximation to maximal indegree.
Following this work, there are two major directions. One is to characterize incentive-compatible peer selection mechanisms with axioms, which is initiated by~\citeauthor{holzman2013impartial}~\shortcite{holzman2013impartial}. %focused on the axiomatic exploration of IC.
% ~\cite{alon2011sum} proposed a 2-partition approach to achieve IC. They randomly divided agents into different groups, and each agent can only influence the selection chance of agents in other groups, which prevents the conflict of interest.
%Besides, they also introduced the approximation ratio (i.e., the ratio between the indegree of selected agents and the maximal indegree) to measure the quality of IC mechanisms.
\citeauthor{mackenzie2015symmetry}~\shortcite{mackenzie2015symmetry} continued this study by adding symmetric axiomizations. The other direction is to improve the approximation ratios of the existing incentive-compatible peer selection mechanisms. \citeauthor{fischer2014optimal}~\shortcite{fischer2014optimal} extended the idea of the partition mechanism to a permutation mechanism, which achieves the optimal approximation ratio ($1/2$) for selecting one agent with in-degree. Then, %and observed the optimal solution when selecting only one single agent.
\citeauthor{bousquet2014near}~\shortcite{bousquet2014near} characterized a class of networks where the permutation mechanism selects an agent close to the optimal. %made further advancement for instances of extremely high indegree.
\citeauthor{bjelde2017impartial}~\shortcite{bjelde2017impartial} generalized the permutation mechanism for selecting multiple agents, and gave both lower and upper bounds of the approximation ratio. % showed deterministic mechanisms can achieve valid approximation bounds by relaxing the exactness,
Recent studies also started to consider an alternative evaluation called additive approximation, which focuses on the expected difference from the optimal rather than the worst-case ratio~\cite{caragiannis2021impartial,caragiannis2022impartial,DBLP:conf/sigecom/CembranoFHK22}.
% \cite{kahng2018ranking} extended the 2-partition mechanism to the $k$-partite mechanism. % Recent research mostly considered the IC selection designs with additive approximation guarantees~\cite{caragiannis2021impartial,caragiannis2022impartial,DBLP:conf/sigecom/CembranoFHK22}.
There are also extensions on the networks, including a weighted network, where the influential power is a weighted in-degree~\cite{kurokawa2015impartial,wang2018tsp,babichenko2020incentive1}, and rank aggregation, where each agent assigns a rank to others~\cite{kahng2018ranking,DBLP:conf/ijcai/MatteiTZ20}.
% Inspired by the ideas of partition and dollar division,~\cite{aziz2016strategyproof} designed an IC selection mechanism (Dollar Partition) for peer review. As one of the most related applications, peer review has attracted various work~\cite{wang2018tsp,babichenko2020incentive1,DBLP:conf/ijcai/MatteiTZ20,lev2021peer}.

There is also a rich body of work focusing on different measurements of influential power. \citeauthor{ghalme2018design}~\shortcite{ghalme2018design} designed a naturally strategy-proof score function to measure the popularity of agents, thereby simplifying the selection. In contrast, % For example, 
\citeauthor{babichenko2018incentive}~\shortcite{babichenko2018incentive} targeted a PageRank-like centrality~\cite{page1999pagerank} and offered a two-path mechanism which achieves a good approximation ratio of $2/3$. Moreover, \citeauthor{babichenko2020incentive2}~\shortcite{babichenko2020incentive2} focused on the progeny of the selected agent and proposed a mechanism with an approximation ratio of $1/(4\ln2)$ in forests. \citeauthor{zhang2021incentive}~\shortcite{zhang2021incentive} then proposed a geometric mechanism with an approximation ratio of $1/2$ in DAGs, which is what we follow here.

\section{Preliminaries} 
%For $n\in\mathbb{N}^{+}$,
Let $\mathcal{G}_{n}$ be the set of all directed acyclic graphs (DAGs) with $n$ nodes and $\mathcal{G}=\bigcup_{n\in\mathbb{N}^{+}}\mathcal{G}_{n}$ be the set of all DAGs. Consider a network represented by a graph $G=(N, E)\in\mathcal{G}$, where $N = \{1,2,\dots, n\}$ is the node set and $E$ is the edge set. Each node $i\in N$ represents an agent, and each edge $(i,j)\in E$ represents that agent $i$ follows (votes for, or quotes) agent $j$. Let $E_i = \{(i,j) \mid (i,j)\in E, j\in N \}$ be the set of all edges from $i$. We say an agent $j$ is influenced by the agent $i$ if there exists a path from $j$ to $i$ in $G$. Let $P(i,G)$\footnote{or simply $P(i)$ if no ambiguity. Additionally, for the sake of tidiness, the notation of the set also indicates its size if no ambiguity.} be the set of agents who are influenced by agent $i$ (including $i$ herself), %The size of set $P(i,G)$
which is referred to as the progeny of agent $i$.

Our goal is to select a group of agents in the network as delegates with larger progeny. Let $\mathcal{S}_k = \{ S \mid S\subseteq N, |S| = k\}$ be the set of all subsets with $k$ agents, and $\mathcal{S}_{\leq k} = \bigcup_{t=0}^k \mathcal{S}_t$. A \emph{$k$-selection mechanism} decides how to choose up to $k$ agents as delegates.

\begin{definition}
A $k$-selection mechanism for $\mathcal{G}$ is a family of functions $f: \mathcal{G}_n \rightarrow [0,1]^{\mathcal{S}_{\leq k}}$ for all $n\in \mathbb{N}^+$, that maps each graph to a probability distribution on all subsets with no more than $k$ agents.
\end{definition}

For a given graph $G\in \mathcal{G}$ and a $k$-selection mechanism $f$, denote $x_S(G) = (f(G))_S$ as the probability of the subset $S\in \mathcal{S}_{\leq k}$ being selected, and $x_i(G) = \sum_{S\in \mathcal{S}_{\leq k} \text{ and } i\in S} x_S(G)$ as the probability of the agent $i$ being selected.

For an agent $i\in N$, she wants her probability to be selected ($x_i$) as large as possible, while for the owner of the mechanism, we want the influential power of the selected group (the sum of the progeny) as large as possible. Unfortunately, if we simply choose an optimal subset with $k$ agents, i.e., $S^*_k \in \arg\max_{S\in \mathcal{S}_k} \sum_{i\in S} |P(i, G)|$, then agents will have incentives to hide their edges to increase their ranks to be selected. We want to avoid such a manipulation, which requires the mechanism to be \emph{incentive-compatible}.

\begin{definition}
A $k$-selection mechanism for $\mathcal{G}$ is incentive-compatible (IC) if for every $n\in \mathbb{N}^+$, and every two graphs $G = (N, E)$, $G' = (N, E')\in \mathcal{G}_n$, such that $E\setminus E_i = E'\setminus E'_i$ and $E_i\supseteq E_i'$ for $i\in N$, we have $x_i(G) \geq x_i(G')$.
\end{definition}

Intuitively, incentive compatibility implies that no matter how other agents follow each other, it is an undominated strategy for any agent not to hide her out-edges. Since an incentive-compatible $k$-selection mechanism cannot always choose a group with the highest influential power, we seek the approximation of the optimum, which guarantees a worst-case ratio between the expected progeny of the selected group and an optimal group for all DAGs.

\begin{definition}
An incentive-compatible $k$-selection mechanism is $\alpha$-optimal if
\[ \inf_{G\in \mathcal{G}} \frac{\mathbb{E}_{S\sim x_S(G)}[\sum_{i\in S} P(i,G)]}{\sum_{i\in S_k^*} P(i,G)} \geq \alpha. \]
\end{definition}

For convenience, we can characterize an optimal group by defining a strict order of agents as follows.

\begin{definition}
For a graph $G = (N,E)\in \mathcal{G}$, and agents $i$, $j\in N$, $i\neq j$, we say $i\succ j$ if either $P(i,G) > P(j,G)$ or $P(i,G) = P(j,G)$ with $i > j$.
\end{definition}

Let $i_t^*$ be the agent with rank $t$ such that $|\{j\mid j\succ i_t^*\}| = t-1$, which must be unique since the order is strict. Then, we can order all the agents as the \emph{ranking sequence} $i_1^*\succ i_2^*\succ \dots \succ i_n^*$. Apparently, $\{i_1^*, \cdots, i_k^* \}$ is an optimal set for selecting $k$ agents. Hence, for our strategic setting, we will pay attention to agents who can pretend to be the first $k$ agents in the ranking sequence by hiding their out-edges.

\begin{definition}\label{def:topset}
For a graph $G = (N, E)$, an agent $i$ belongs to the $k$-influential set $S^{\mathsf{inf.}}_k(G)$ if $|\{j\mid j\succ i\}| < k$ holds in the graph $G' = (N, E\setminus E_{i})$.
\end{definition}

In this paper, we mainly focus on the cases for $k \in \{1, 2\}$. Hence, we use some observations about $S^{\mathsf{inf.}}_1(G)$ and $S^{\mathsf{inf.}}_2(G)$. To make it easy to follow, we present the observations about $S^{\mathsf{inf.}}_1(G)$ below as preparation and present the observations of $S^{\mathsf{inf.}}_2(G)$ in Section~\ref{sec:ontop2-set}.

\begin{observation}[\cite{zhang2021incentive}]\label{ob:1-like}
For any graph $G$, the set $S^{\mathsf{inf.}}_1(G)$ can be written as $\{i_1, i_2, \cdots, i_m\}$, where $m\geq 1$, $i_1 = i_1^*$, and $i_{t+1} \in P(i_t)\setminus\{i_t\}$ for all $t<m$.
\end{observation}

Intuitively, the agent $i_1^*$ who ranks the first is naturally in the 1-influential set. Furthermore, if there are more than one agent in the set, the agent who has a lower rank must be in the progeny of those who have higher ranks; otherwise, she will still have a lower rank when deleting her out-edges. Hence, in other words, we can find a path in $G$ that passes through all agents in the set with the order of their ranks.

\begin{observation}[\cite{zhang2021incentive}]\label{ob:1-ratio}
For any graph $G$, if the set $S^{\mathsf{inf.}}_1(G)=\{i_1, i_2, \cdots, i_m\}$ has more than one agent, i.e., $m>1$, then for any $1<t\leq m$, $2P(i_t) \geq P(i_1)$.
\end{observation}

Inferred from Observation~\ref{ob:1-like}, if there is another agent except for $i_1^*$ is in the 1-influential set, she should hold at least half of $i_1^*$'s progeny to make herself rank the first after removing her out-edges.

\section{Select One Agent}\label{section:3}
In this section, we present our result for only selecting $k=1$ agent. Not only does it help us understand the proposed methods in the following section for $k = 2$, but also fills the gap between the existing mechanisms and the upper bound of approximation ratios for IC 1-selection mechanisms, which is $1/(1+\ln 2)$ confirmed by~\citeauthor{zhang2021incentive}~\shortcite{zhang2021incentive}.

Formally, our method can be viewed as a general variant of the modified Babichenko's mechanism~\shortcite{babichenko2020incentive2}. %that selects one single agent in forests. %By assigning a fixed probability to the last candidate, we prove that this modified $1$-selection mechanism exactly meets the optimal approximation ratio ($1/(1+\ln 2)$ confirmed by~\cite{zhang2021incentive}). 

\begin{framed}
% \label{frame:bb}
 \noindent\textbf{$\beta$-logarithmic Mechanism ($\beta$-LM)}
 
 \noindent\rule{\textwidth}{0.5pt}
 \begin{enumerate}
     \item Given a network $G=(N,E)$, find the 1-influential set $S^{\mathsf{inf.}}_1(G)=\{i_{1},\dots,i_{m}\}$, where $i_{t} \succ i_{t+1}$ for all $1 \leq t < m$.
     \item Assign the probability of each agent to be selected as follows:
     \begin{align*}
         x_j = \begin{cases}
         %\frac{1}{1+\ln 2},
         \beta, & j = i_m \\
         %\frac{\ln 2}{1+\ln 2}
         (1-\beta) \log_2 \frac{P(i_t)}{P(i_{t+1})}, & j = i_t, t\neq m\\
         0, & j \notin S^{\mathsf{inf.}}_1(G).
         \end{cases}
     \end{align*}
 \end{enumerate}
\end{framed}

The total probability of selecting one agent in $\beta$-LM is at most $\beta + (1-\beta)\log_2(P(i_1)/P(i_m)) \leq 1$ by the fact in Observation~\ref{ob:1-ratio}. Hence, the probabilities assigned by the mechanism are valid as long as $0\leq \beta \leq 1$. Then, we show the mechanism is IC when $\beta \geq 1/2$.

% \begin{lemma}\label{lemma1}
% For any graph $G=(N,E)$, and the set $S^{\mathsf{inf.}}_1(G)=\{i_{1},\dots,i_{m}\}$, we have $i_m\in S^{\mathsf{inf.}}_1((N, E\setminus E_{i_t}))$ for all $1\leq t\leq m$.
% \end{lemma}

% \begin{proof}
% For any agent $i_t\in S^{\mathsf{inf.}}_1(G)$,
% % This can be proved in two steps. Let the 1-influential set be $S^{\mathsf{inf.}}_1(G)=\{i_{m},\dots,i_{1}\}$ for a network $G=(N,E)$, where $i_1\succ i_2 \dots \succ i_{m-1} \succ i_{m}$. 
% % \begin{enumerate}
% %     \item For candidate $i_{t}\in S^{\mathsf{inf.}}_1(G)$, her strategic action will have influence on candidates $i_{t-1},\dots,i_{1}$ ranking before her, since she can reduce the progeny of them by hiding her out-edges. In the worst case, $i_t$ has power to remove all these candidates from the set in a graph $G'=(N,E\setminus E_{i_{t}})$.
% %     \item For agents outside the 1-influential set in $G$, after the maximum progeny decreases from $P(i_{1},G)$ to $P(i_{t},G')$, a part of them have chances to be candidates in $S^{\mathsf{inf.}}_1(G')$.  
% % \end{enumerate}
% % As any candidate hiding her out-edges may remove the candidates ranking ahead of her or add other agents to the 1-influential set, her ranking order can only be improved or remain unchanged.           
% \end{proof}

\begin{theorem}\label{thm:lm-ic}
A $\beta$-logarithmic mechanism is IC if $\beta \geq 1/2$.
\end{theorem}

\begin{proof}
% We prove the statement from both sides.
% \begin{itemize}
%     \item[``$\Rightarrow$":] If the mechanism is IC,
%     \item[``$\Leftarrow$":]
% \end{itemize}
For any graph $G\in\mathcal{G}$, let $S^{\mathsf{inf.}}_1(G)=\{i_{1},\dots,i_{m}\}$. Then, we consider three different types of agents.
\begin{enumerate}
    \item For an agent $i\notin S^{\mathsf{inf.}}_1(G)$, by definition, she can never pretend to be the agent with rank 1 by hiding her out-edges. Hence, she will always not belong to the 1-influential set and will have 0 probability to be chosen.
    \item For an agent $i_t\in S^{\mathsf{inf.}}_1(G)$ such that $t<m$, no matter how she hides her out-edges, $i_{t+1}$ will always belong to the 1-influential set because $i_{t+1}\in P(i_t)$ (Observation~\ref{ob:1-like}) and her progeny cannot be decreased. Hence, the probability of $i_t$ to be chosen will not change.
    \item For the agent $i_m\in S^{\mathsf{inf.}}_1(G)$, if she hides some of her out-edges, there may happen two cases. (i) If there is no agent in $P(i_m)$ occurs in the new 1-influential set, %the new set will only contain agent $i_m$ and
    the probability of $i_m$ to be chosen will remain $\beta$. (ii) If there exists at least one agent in $P(i_m)$ that occurs in the new 1-influential set, let $i_1'$ be the first agent in the new set. Let $i_{m+1}$ be the one with the highest rank after $i_m$ in the new set. Then, the probability of $i_m$ to be chosen will become $(1-\beta) \log_2 (P(i_m)/P(i_{m+1})) \leq (1-\beta) \log_2 (P(i_1')/P(i_{m+1})) \leq (1-\beta) \leq \beta$ by the fact of $\beta \geq 1/2$ and Observation~\ref{ob:1-ratio}.
\end{enumerate}
Taking all the above together, no agent can increase her probability to be chosen by hiding her out-edges. %Therefore, the mechanism is IC.
% For the last candidate $i_m$, she already has the largest selection probability $1/(1+\ln 2)$ than others in the set. Therefore, she has no intention to hide her out-edges. For other candidates, $i_{m-1},\dots,i_1$, firstly by Lemma~\ref{lemma1}, $i_m$ is still in the 1-influential set and hence they cannot change their probabilities to be $1/(1+\ln 2)$. Besides, they cannot improve the chances to be selected since their probabilities are related to their progeny and their children's progeny, which are independent of their manipulations. For agents outside the set, they cannot become candidates because whether an agent in the 1-influential set is independent of her out-edges according to Definition~\ref{def:topset}. As nobody can improve her selection probability, the mechanism satisfies IC. 
\end{proof}

Now we can compute the approximation ratios of IC $\beta$-LMs, from which we can find that the optimal $\beta$-LM is also an optimal IC selection mechanism for $k=1$.

\begin{theorem}\label{thm:LM-opt}
An IC $\beta$-logarithmic mechanism ($1/2\leq \beta \leq 1$) is $\left(\min\left\{ \frac{1}{2}\left( \beta + \frac{1-\beta}{\ln 2} \right), \beta \right\}\right)$-optimal.
\end{theorem}

\begin{proof}
For any graph $G\in\mathcal{G}$, let $S^{\mathsf{inf.}}_1(G)=\{i_{1},\dots,i_{m}\}$. If $m=1$, then we have
\[ \mathbb{E}_{i\sim x_i}[P(i)] / P(i_1^*) = x_{i_1} P(i_1) / P(i_1) = \beta. \]
%the structure of its 1-influential set can only be a root or a single chain. In the root case, the selection probability for the single candidate $i_m$ is $1/(1+\ln 2)$. We can easily calculate the approximation ratio $\alpha=1/(1+\ln 2)$. In the single-chain case, the expected progeny for a network is
If $m>1$, then we have
\begin{align*}
    &\mathbb{E}_{i\sim x_i}[P(i)] / P(i_1^*) \\
    = &\ \frac{1-\beta}{P(i_1^*)} \sum_{t=1}^{m-1}P(i_{t})\log_{2}\frac{P(i_t)}{P(i_{t+1})} + \beta \frac{P(i_m)}{P(i_1^*)} \\
    = &\ \frac{1-\beta}{\ln{2}\cdot P(i_1^*)} \sum_{t=1}^{m-1}P(i_{t})\int_{P(i_{t+1})}^{P(i_{t})}\frac{\mathrm{d}z}{z} + \beta \frac{P(i_m)}{P(i_1^*)} \\
    \geq &\ \frac{1-\beta}{\ln{2}\cdot P(i_1^*)} \sum_{t=1}^{m-1}\int_{P(i_{t+1})}^{P(i_{t})}\mathrm{d}z + \beta \frac{P(i_m)}{P(i_1^*)} \\
    = &\ \frac{1-\beta}{\ln{2}\cdot P(i_1)} \sum_{t=1}^{m-1}(P(i_{t})-P(i_{t+1})) + \beta \frac{P(i_m)}{P(i_1)} \\
    = &\ \frac{1-\beta}{\ln 2} + \left( \beta - \frac{1-\beta}{\ln 2} \right) \frac{P(i_m)}{P(i_1)} \\
    \geq &\ \frac{1}{2}\left( \beta + \frac{1-\beta}{\ln 2} \right).
\end{align*}

Therefore, the mechanism is $\left(\min\left\{ \frac{1}{2}\left( \beta + \frac{1-\beta}{\ln 2} \right), \beta \right\}\right)$-optimal.
\end{proof}

It is not hard to find out that when $\beta = 1/(1 + \ln 2)$, the value $\min\left\{ \frac{1}{2}\left( \beta + \frac{1-\beta}{\ln 2} \right), \beta \right\}$ takes it maximum as $1/(1 + \ln 2)$, i.e., the optimal $\beta$-LM is $(1/(1 + \ln 2))$-LM, which is $(1/(1 + \ln 2))$-optimal. Recall that \citeauthor{zhang2021incentive}~\shortcite{zhang2021incentive} has proved that no IC and fair\footnote{Fairness is a quite weak property for single agent selection that only requires $i_1^*$ has the same probability to be chosen if the 1-influential set and the structure formed by $P(i_1^*)$ remain the same. It is clear to see that $\beta$-LM satisfies this fairness.} selection mechanism can be $\alpha$-optimal with $\alpha > 1/(1+\ln 2)$. Hence, we can infer the optimality of $(1/(1 + \ln 2))$-LM.

\begin{corollary}
There is no other IC and fair selection mechanism for $k=1$ that can have a higher approximation ratio than $(1/(1 + \ln 2))$-LM.
\end{corollary}

At the end of this section, we give a running example of $(1/(1 + \ln 2))$-LM.

\begin{example}\label{eg1}
Consider the network depicted in Figure~\ref{fig:top1selection}, where $S^{\mathsf{inf.}}_1(G)=\{i_{1}, i_{2}, i_3, i_4\}$. %$i_1\succ i_2 \succ i_3 \succ i_4 \succ j$. Only agents $i_1$, $i_2$, $i_3$, and $i_4$ have the largest progeny after hiding their out-edges. These agents form the 1-influential set (represented by a dashed border).
For the last agent $i_4$ in the set, her selection probability $x_{i_4}$ is $1/(1+\ln2)\approx0.59$. For the agents $i_3$, $i_2$ and $i_1$, their selection probabilities are

\begin{align*}
    x_{i_3}&=\frac{\ln2}{1+\ln2}\log_2\frac{P(i_3)}{P(i_4)}=\frac{\ln2}{1+\ln2}\log_2\frac{5}{4}\approx0.13; \\
    x_{i_2}&=\frac{\ln2}{1+\ln2}\log_2\frac{P(i_2)}{P(i_3)}=\frac{\ln2}{1+\ln2}\log_2\frac{6}{5}\approx0.11; \\
    x_{i_1}&=\frac{\ln2}{1+\ln2}\log_2\frac{P(i_1)}{P(i_2)}=\frac{\ln2}{1+\ln2}\log_2\frac{7}{6}\approx0.09.
\end{align*}

% Similarly, the selection probabilities of agents $i_2$ and $i_1$ are $x_{i_2}\approx0.11$ and $x_{i_1}\approx0.09$. 
\end{example}

\begin{figure}[!htbp]
    \centering
    \includegraphics[width=0.48\textwidth]{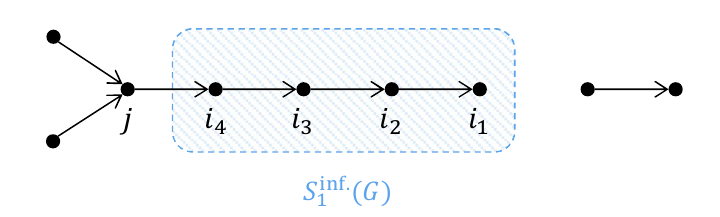}
    \caption{An example of the network, where the marked agents have the relationship as $i_1\succ i_2 \succ i_3 \succ i_4 \succ j$. The 1-influential set is represented by a dashed border.}
    \label{fig:top1selection}
\end{figure}

\section{Select Two Agents}
We start to consider selecting up to $k=2$ agents as delegates.
% When at most two agents with larger progeny are required, e.g., a social media platform is evaluating the top two influential users for advertising,
To select agents with larger progeny, one possible
%natural extension of the previous
approach is to find the second delegate from the 1-influential set as well. However, this limits the performance of IC mechanisms.
%the approximation ratio of this extended version can not exceed $1/2$.

\subsection{Limitation of the 1-influential Set}
The limitation of selecting agents from the 1-influential set for $k=2$ mainly comes from the fact that there may be only a single agent in the set.

\begin{theorem}\label{thm:limit}
If an IC 2-selection mechanism only selects agents in the 1-influential set, then it cannot be $\alpha$-optimal with $\alpha > 1/2$.
% The series of IC mechanisms that select two agents from the 1-influential set cannot have an approximation ratio greater than $1/2$.  
\end{theorem}

\begin{proof}
Consider a two-star graph shown in Figure~\ref{fig:twostar}. Suppose $P(i_1) = P(j) = y$ and $i_1 > j$. Then the 1-influential set of this graph $S^{\mathsf{inf.}}_1(G)$ only contains $i_1$. 
% (the one with larger progeny or the one with higher lexicographic order for tie-breaking).
Therefore, even if the mechanism can always select agent $i_1$ with probability $1$, the approximation ratio in this graph is only $(1\cdot y)/(y+y) = 1/2$. Hence, if only selecting agents in the 1-influential set, an IC 2-selection mechanism cannot achieve an approximation ratio of more than $1/2$.
% Notice that the optimal result includes the two agents, $i_1$ and $j$, but the extended approach could only assign selection probability to the single candidate $i_1$ in the 1-influential set. Therefore, in this case, the highest approximation ratio is $1/2$ by giving all probability to $i_1$, i.e., $x_{i_1}=1$.     
\end{proof}

\begin{figure}[!htbp]
    \centering
    \includegraphics[width=0.48\textwidth]{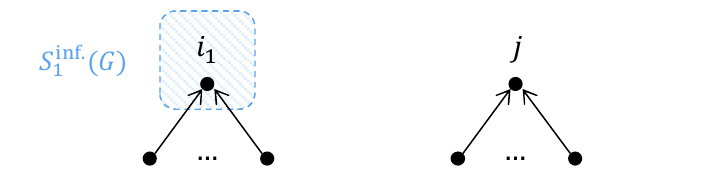}
    \caption{A two-star network where the hubs satisfy $i_1\succ j$. Thus, the 1-influential set only contains $i_1$.}
    \label{fig:twostar}
\end{figure}

We can also show that the limitation described in Theorem~\ref{thm:limit} is tight by providing the following mechanism.
%This fixed selection progress motivates us to provide one deterministic IC mechanism to select two influential agents. We find that the proposed mechanism achieves the approximation ratio of $1/2$.

\begin{framed}
% \label{frame:bb}
 \noindent\textbf{Least Deterministic Mechanism (LDM)}
 
 \noindent\rule{\textwidth}{0.5pt}
 \begin{enumerate}
     \item Given a network $G=(N,E)$, find the 1-influential set $S^{\mathsf{inf.}}_1(G)=\{i_{1},\dots,i_{m}\}$, where $i_{t} \succ i_{t+1}$ for all $1 \leq t < m$.
     \item Assign the probability of each agent to be selected as follows: %The mechanism gives the selection probability distribution on all agents as the following.
     \begin{align*}
         x_j = \begin{cases}
         1, & j = i_m, \text{or } j = i_{m-1} \\
         0, & j = i_t, t < m-1, \text{or } j \notin S^{\mathsf{inf.}}_1(G).\\
         \end{cases}
     \end{align*}
 \end{enumerate}
\end{framed}

Intuitively, LDM deterministically selects the last two agents in the 1-influential set or it selects the only agent in the set if $|S^{\mathsf{inf.}}_1(G)| = 1$.

\begin{example}
We take the networks shown in Figure~\ref{fig:top1selection} and Figure~\ref{fig:twostar} as running examples. In Figure~\ref{fig:top1selection}, there are four agents, $i_1$, $i_2$, $i_3$, and $i_4$, in the 1-influential set. Hence, LDM deterministically selects the last two agents $i_4$ and $i_3$, i.e., $x_{i_4}=x_{i_3}=1$. In Figure~\ref{fig:twostar}, there is a single agent $i_1$ in the 1-influential set. Hence, LDM deterministically selects the agent $i_1$ only.
\end{example}

Now we prove that LDM is incentive-compatible and $1/2$-optimal as follows.

\begin{theorem}
LDM is an IC 2-selection mechanism.
\end{theorem}

\begin{proof}
For any graph $G\in \mathcal{G}$, suppose that $S^{\mathsf{inf.}}_1(G)=\{i_{1},\dots,i_{m}\}$. Then we consider the following cases. 
% There are at most two candidates in the 1-influential set that have chances to be selected. 
\begin{enumerate}
    \item For an agent $i\notin S^{\mathsf{inf.}}_1(G)$, same as the first point in the proof of Theorem~\ref{thm:lm-ic}, she always has no chance to be chosen by hiding her out-edges.
    \item If $m\leq 2$, the agents in $S^{\mathsf{inf.}}_1(G)$ will be deterministically selected. Hence, they have no incentive to hide their out-edges.
    \item If $m > 2$, first, for agents $i_m$ and $i_{m-1}$ who will be deterministically selected, they have no incentive to hide their out-edges. Then, for any agent $i_t\in S^{\mathsf{inf.}}_1(G)$ with $i<m-1$, no matter how she hides her out-edges, $i_{m}$ and $i_{m-1}$ will always belong to the 1-influential set because $\{i_m, i_{m-1}\}\subseteq P(i_t)$ (Observation~\ref{ob:1-like}) and their progeny cannot be decreased. Hence, the probability of $i_t$ to be chosen will remain to be 0.
\end{enumerate}
Taking all the above together, no agent can increase her probability to be chosen by hiding her out-edges. Therefore, the mechanism is IC.
% If $|S^{\mathsf{inf.}}_1(G)|=1$, the only one candidate $i_m$ will be selected fixedly. She has no desire to change this beneficial result. If $|S^{\mathsf{inf.}}_1(G)|\geq2$, candidate $i_m$ and $i_{m-1}$ will be selected fixedly and thus they are not required to take any strategic actions. For other candidates in the 1-influential set, by Lemma~\ref{lemma1}, even they hide their out-edges, candidate $i_m$ and $i_{m-1}$ are still in the set. So they cannot improve their selection probabilities from 0 to 1. For agents outside the set, they cannot become candidates by Definition~\ref{def:topset}. Therefore, Mechanism 2 satisfies IC.  
\end{proof}

\begin{theorem}
LDM is $1/2$-optimal.
\end{theorem}

\begin{proof}
For any graph $G\in\mathcal{G}$, if $|S^{\mathsf{inf.}}_1(G)|=1$, then LDM deterministically selects the agent $i_1^*$. Therefore, in this case, we have
\[ \frac{\mathbb{E}_{S}[\sum_{i\in S} P(i)]}{\sum_{i\in S_2^*} P(i)} = \frac{P(i_1^*)}{P(i_1^*) + P(i_2^*)} \geq \frac{P(i_1^*)}{P(i_1^*) + P(i_1^*)} = \frac{1}{2}. \]
%candidate $i_m$ will have a selection probability $x_{i_m}=1$. Since the optimal solution contains two agents with largest progeny and candidate $i_m$ is one of them, the approximation ratio $\alpha\leq1/2$. 
If $|S^{\mathsf{inf.}}_1(G)|\geq 2$, then LDM deterministically selects the agent $i_m$ and $i_{m-1}$. By Observation~\ref{ob:1-ratio}, in this case, we have
\begin{align*}
    \frac{\mathbb{E}_{S}[\sum_{i\in S} P(i)]}{\sum_{i\in S_2^*} P(i)} & = \frac{P(i_m)+P(i_{m-1})}{P(i_1^*) + P(i_2^*)} \\ 
    & \geq \frac{P(i_m)+P(i_{m-1})}{2P(i_1)} \\
    & = \frac{1}{2}\left( \frac{P(i_m)}{P(i_1)} + \frac{P(i_{m-1})}{P(i_1)} \right) \\
    & \geq \frac{1}{2}\left( \frac{1}{2} + \frac{1}{2} \right) = \frac{1}{2}.
\end{align*}
% \[ \frac{\mathbb{E}_{S}[\sum_{i\in S} P(i)]}{\sum_{i\in S_2^*} P(i)} = \frac{P(i_m)+P(i_{m-1})}{P(i_1^*) + P(i_2^*)} \geq \frac{P(i_m)+P(i_{m-1})}{P(i_1^*) + P(i_1^*)} = \frac{1}{2}. \]
%The optimal agents here are the top one candidate $i_1$ with another agent $j$ who satisfies $P(j^{*},G)\leq P(i_1^{*},G)$. In this case, we have $\mathbb{E}_{i\sim x_i}[P(i)]=P(i_m)+P(i_{m-1})\geq P(i_1^{*})$. The optimal solution is $P(i_1^{*})+P(j^{*})\leq 2P(i_1^{*})$. So the approximation ratio $\alpha\leq1/2$. Thus we complete the proof.
Therefore, LDM is $1/2$-optimal.
\end{proof}

If we consider general IC 2-selection mechanisms, we may have a higher upper bound of the approximate ratio. This suggests the limitation of only selecting agents from the 1-influential set when we target up to two delegates.
%now prove an upper bound for any IC selection mechanisms that output at most two influential agents.

\begin{theorem}
There is no IC 2-selection mechanism that can be $\alpha$-optimal with $\alpha > 23/27$.
\end{theorem}

\begin{proof}
Consider three networks with four agents shown in Figure~\ref{fig:upperbound}. Applying a generic IC 2-selection mechanism on these graphs, suppose the probabilities of each agent $i$ being chosen in the three graphs are $x_i^{(a)}$, $x_i^{(b)}$ and $x_i^{(c)}$. Notice that network (b) can be obtained by agent 2 or 4 in network (a) hiding their out-edges (corresponding to agent 2 or 1 in network (b)), while network (c) can be obtained by agent 3 hiding her out-edge (corresponding to agent 1 or 3 in network (c)). Since the selection mechanism is IC, we have the following constraints:

\begin{align}
    x_2^{(b)}\leq x_2^{(a)},\quad & x_1^{(b)}\leq x_4^{(a)}; \\
    x_3^{(c)}\leq x_3^{(a)},\quad & x_1^{(c)}\leq x_3^{(a)}.
\end{align}

Moreover, the mechanism selects at most 2 agents. Hence,
\begin{align}
    \sum_{i=1}^4 x_i^{(a)} \leq 2, \quad \sum_{i=1}^4 x_i^{(b)} \leq 2, \quad \sum_{i=1}^4 x_i^{(c)} \leq 2.
\end{align}

The approximation ratio of the mechanism must be no more than the least ratio in these three graphs, i.e.,

\begin{equation*}
    \begin{split}
        \alpha \leq \min & \left\{ \frac{4x_1^{(a)} + 3x_2^{(a)} + 2x_3^{(a)} + x_4^{(a)}}{7}, \right. \\
    & \quad \frac{x_1^{(b)} + 3x_2^{(b)} + 2x_3^{(b)} + x_4^{(b)}}{5}, \\
    & \quad \left. \frac{2x_1^{(c)} + x_2^{(c)} + 2x_3^{(c)} + x_4^{(c)}}{4} \right\}
    \end{split}
\end{equation*}

With constraints (1) - (3), we can calculate the highest value of the minimum. Therefore, we have $\alpha\leq 23/27$ and the equations holds when

\begin{align*}
    & x_1^{(a)} = 2/3, x_2^{(a)} = 17/27, x_3^{(a)} = 19/27, x_4^{(a)} = 0; \\
    & x_1^{(b)} = 0, x_2^{(b)} = 17/27, x_3^{(b)} = 1, x_4^{(b)} = 10/27; \\
    & x_1^{(c)} = 19/27, x_2^{(c)} = 16/27, x_3^{(c)} = 19/27, x_4^{(c)} = 0.
\end{align*}

\end{proof}

\begin{figure}[htbp]
    \centering
    \includegraphics[width=0.48\textwidth]{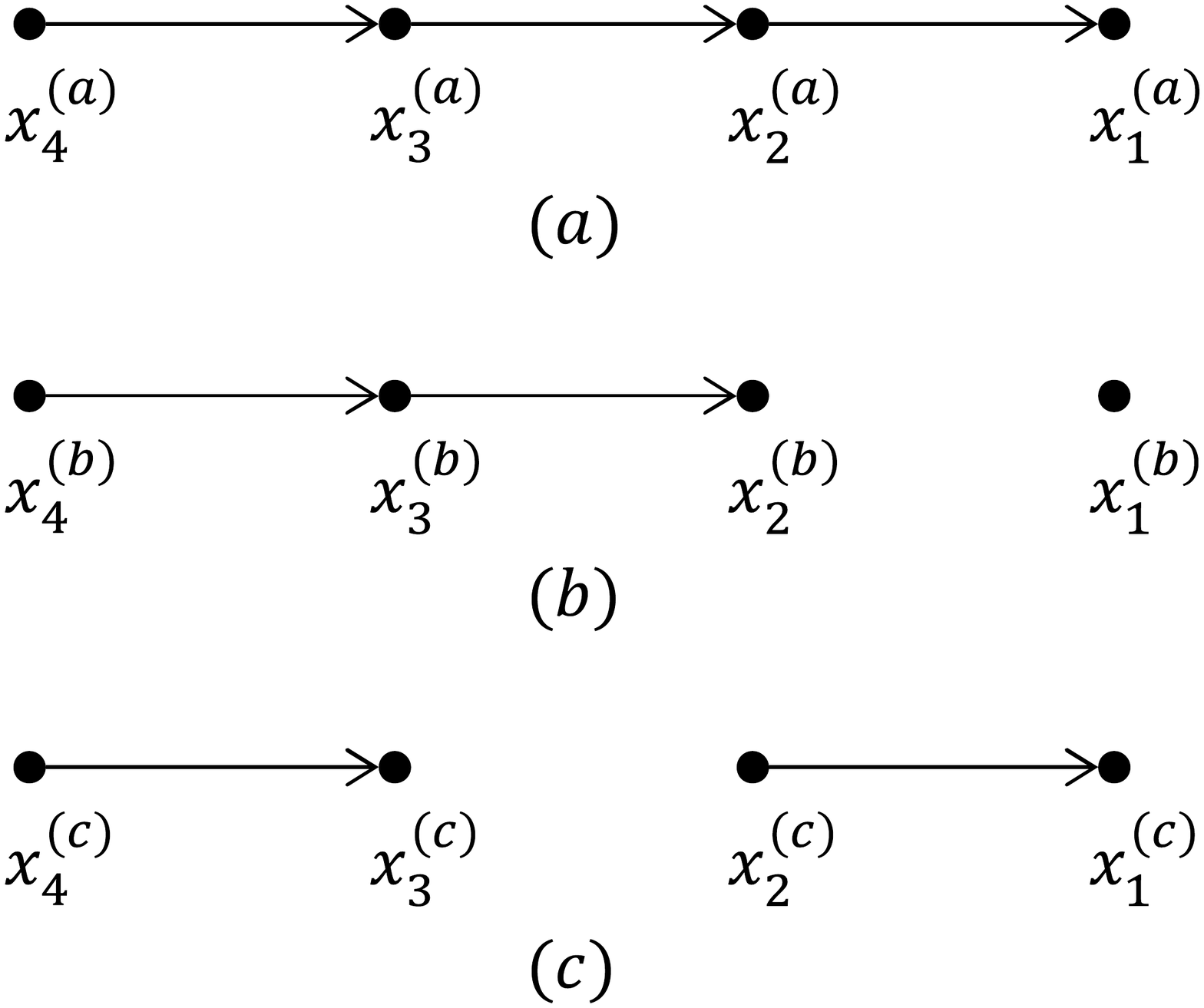}
    \caption{Three networks with four agents where (b) and (c) can be obtained by one of the agents in (a) hiding her out-edge. The probabilities of each agent being chosen by a generic IC 2-selection mechanism are attached beside the node.}
    \label{fig:upperbound}
\end{figure}

\subsection{Utilizing the 2-influential Set}\label{sec:ontop2-set}
To break through the limitation of the 1-influential set, one natural idea is to consider the 2-influential set. We first characterize the set by following observations.
% We propose another IC mechanism to select the top two influential agents. This mechanism is more efficient, and the approximation ratio is closer to the upper bound than Mechanism 2. 

\begin{observation}\label{ob:1in2}
For any graph $G$, $S^{\mathsf{inf.}}_1(G)\subseteq S^{\mathsf{inf.}}_2(G)$.
\end{observation}

\begin{observation}\label{ob:12in2}
For any graph $G$, $\{ i_1^*, i_2^* \}\subseteq S^{\mathsf{inf.}}_2(G)$.
\end{observation}

According to Definition~\ref{def:topset}, agents who can pretend to be the first in the ranking sequence are definitely in the 2-influential set. The agents $i_1^*$ and $i_2^*$ who rank first and second are also naturally in the 2-influential set. Based on the relationship between $i_1^*$ and $i_2^*$, the 2-influential set will have different forms.

\begin{observation}\label{ob:2-chain}
For any graph $G$, if $i_2^*\in P(i_1^*)$, the set $S^{\mathsf{inf.}}_2(G)$ can be written as $\{ i_1, i_2, \cdots, i_m \}$, where $i_1 = i_1^*$, $i_2 = i_2^*$, and $i_{t+1}\in P(i_t)$ for all $t<m$.
\end{observation}

\begin{proof}
When there is another agent $i_3$, except for $i_1^*$ and $i_2^*$, is in the 2-influential set, she must have the ability to decrease $i_1^*$ or $i_2^*$'s progeny by hiding her out-edges, i.e., at least one of $i_3\in P(i_1^*)$ and $i_3\in P(i_2^*)$ is satisfied. If $i_2^*\in P(i_1^*)$, we will show $i_3$ must belong to $P(i_2^*)$ by contradiction as follows.

If $i_3\notin P(i_2^*)$, then she cannot decrease $i_2^*$'s progeny by hiding her out-edges. Since $i_2^*\notin P(i_3)$, $i_2^*\in P(i_1^*)$ is always satisfied no matter how $i_3$ hides her out-edges. Hence, in order to rank first or second after hiding out-edges, $i_3\succ i_2^*$ must be satisfied, which contradicts that $i_2^*$ is with rank 2. Therefore, $i_3\in P(i_2^*)$.

% Similarly, when there is another agent $i_4\prec i_3$ in the 2-influential set, she must belong to $P(i_3)$, and so forth.

Similarly, when there is another agent $i_4\prec i_3$ in the 2-influential set, she must belong to $P(i_3)$, and this pattern extends to subsequent agents in the 2-influential set. %, such as $i_5$, $i_6$ and beyond. 
\end{proof}

\begin{observation}\label{ob:2-single}
For any graph $G$, if $i_2^*\notin P(i_1^*)$, the set $S^{\mathsf{inf.}}_2(G)$ can be written as $\{i_1, i_2, \cdots, i_m\}$, where $i_1 = i_1^*$, $i_2 = i_2^*$, and for others, at least one of the below is true:
\begin{itemize}
    \item $i_3\in P(i_1)$, $i_{t+1}\in P(i_t)$ for all $3\leq t<m$;
    \item $i_3\in P(i_2)$, $i_{t+1}\in P(i_t)$ for all $3\leq t<m$.
\end{itemize}
\end{observation}

\begin{proof}
When there is another agent $i_3$, except for $i_1^*$ and $i_2^*$, is in the 2-influential set, at least one of $i_3\in P(i_1^*)$ and $i_3\in P(i_2^*)$ is satisfied. %to have the ability to decrease their progeny.

We first assume $i_3\in P(i_1^*)$. Then, when there is another agent $i_4\prec i_3$ in the 2-influential set, we show that $i_4\in P(i_3)$ by contradiction as follows. If $i_4\notin P(i_3)$, she cannot decrease $i_3$'s progeny by hiding her out-edges. Since $i_3\notin P(i_4)$, $i_3\in P(i_1^*)$ is always satisfied no matter how $i_4$ hides her out-edges. Hence, $i_4\succ i_3$ must be satisfied to make $i_4$ be in the 2-influential set, which makes a contradiction. Therefore, $i_4\in P(i_3)$. Similarly, when there is another agent $i_5\prec i_4$ in the 2-influential set, she must belong to $P(i_4)$, and %so forth
this pattern extends to subsequent agents in the 2-influential set. The same results can also be obtained similarly if $i_3\in P(i_2^*)$.
\end{proof}

From the above observations, we can see that the structure of the 2-influential set is much more complex than that of the 1-influential set. Furthermore, the progeny of the last agent in the 2-influential set might be much smaller than $P(i_1^*)$ since $P(i_2^*)$ might be small. These are the main difficulties to utilize the 2-influential set. Extending the idea of LDM and $\beta$-LM, we propose the following mechanism.

\begin{framed}
% \label{frame:bb}
 \noindent\textbf{Logarithm After Least Deterministic (LALD)}
 
 \noindent\rule{\textwidth}{0.5pt}
 \begin{enumerate}
    \item Given a network $G=(N,E)$, find the 1-influential set $S^{\mathsf{inf.}}_1(G)$ and the 2-influential set $S^{\mathsf{inf.}}_2(G)$.
     \item If $S^{\mathsf{inf.}}_2(G)\setminus S^{\mathsf{inf.}}_1(G) = \emptyset$, then $S^{\mathsf{inf.}}_1(G) = S^{\mathsf{inf.}}_2(G) = \{i_1, i_2 \dots, i_m\}$, where $i_t\succ i_{t+1}$ for all $1\leq t < m$. Then, assign the probability of each agent to be selected as follows:
     \begin{align*}
         x_j = \begin{cases}
         1, & j = i_{m} \\
         \frac{1}{1+\ln 2}, & j = i_{m-1} \\
         \frac{\ln 2}{1+\ln 2} \log_2 \frac{P(i_t)}{P(i_{t+1})}, & j = i_t, t<m-1 \\
         0, & j \notin S^{\mathsf{inf.}}_2(G).
         \end{cases}
     \end{align*}
     \item If $S^{\mathsf{inf.}}_2(G)\setminus S^{\mathsf{inf.}}_1(G) \neq \emptyset$, suppose $S^{\mathsf{inf.}}_2(G) = \{i_1, \dots, i_m\}$ where $i_t\succ i_{t+1}$ for all $1\leq t < m$. First, deterministically select the agent $i_m$, i.e., $x_{i_m} = 1$. Then, select the second agent by applying $(1/(1+\ln2))$-LM on $G$.
 \end{enumerate}
\end{framed}

Intuitively, LALD first deterministically selects the last agent in the 2-influential set. Then, it uses the same probability distribution as $(1/(1+\ln2))$-LM to select another agent among the remaining agents in the 1-influential set.

\begin{example}
We take the networks shown in Figure~\ref{fig:top1selection} and Figure~\ref{fig:lald} as running examples.

In Figure~\ref{fig:top1selection}, suppose $j<i_2$. Then, $S^{\mathsf{inf.}}_2(G) = S^{\mathsf{inf.}}_1(G) = \{i_1, i_2, i_3, i_4\}$. Hence, LALD first deterministically selects agent $i_4$, i.e., $x_{i_4} = 1$. For the remaining agents in $S^{\mathsf{inf.}}_1(G)$, LALD assigns the probabilities as $x_{i_{3}} = 1/(1+\ln 2)\approx0.59$, $x_{i_2} = \ln2 / (1+\ln 2) \log_2 (P(i_2)/P(i_3)) \approx 0.11$, and $x_{i_1} = \ln2 / (1+\ln 2) \log_2 (P(i_1)/P(i_2)) \approx 0.09$.

In Figure~\ref{fig:lald}, suppose $i_2 > i_3$ and $i_4 > j$. Then, $S^{\mathsf{inf.}}_1(G) = \{i_1\}$ and $S^{\mathsf{inf.}}_2(G) = \{i_1, i_2, i_3, i_4\}$. Hence, LALD first deterministically selects agent $i_4$, i.e., $x_{i_4} = 1$. LALD then runs $(1/(1+\ln 2))$-LM, which assigns the probabilities among $S^{\mathsf{inf.}}_1(G)$. Here, it assigns that $x_{i_1} = 1/(1+\ln 2)\approx 0.59$.
\end{example}

% \begin{example}
% For the network $G$ shown in Figure~\ref{fig:top1selection}, no agent can take the second place after hiding her out-edges, i.e., the network satisfies $S^{\mathsf{inf.}}_2(G)\setminus S^{\mathsf{inf.}}_1(G) = \emptyset$. As the last agent in $S^{\mathsf{inf.}}_1(G)$, agent $i_4$ will be selected with probability $x_{i_4}=1$. For other agents in $S^{\mathsf{inf.}}_1(G)$, according to the probability distribution in $1/(1+\ln2)$-LM, agent $i_{3}$'s selection probability $x_{i_{3}}$ is $1/(1+\ln 2)\approx0.59$, agent $i_2$ and $i_1$'s selection probabilities are still $x_{i_2}\approx0.11$ and $x_{i_1}\approx0.09$.    
% \end{example}

% \begin{example}
% Consider the network $G$ depicted in Figure~\ref{fig:lald}. After hiding out-edges, there are four agents, $i_1$, $i_2$, $i_3$ and $i_4$ in the top two positions, i.e., the 2-influential set $S^{\mathsf{inf.}}_2(G)=\{i_1, i_2, i_3, i_4\}$. Because only agent $i_1$ can have the largest progeny after hiding her out-edges, the 1-influential set $S^{\mathsf{inf.}}_1(G)=\{i_1\}$. In this case, the last agent $i_4$ in $S^{\mathsf{inf.}}_2(G)$ will be fixedly selected with probability $x_{i_4}=1$. Besides, agent $i_1$ in $S^{\mathsf{inf.}}_1(G)$ will have a selection probability of $x_{i_1}=1/(1+\ln2)\approx0.59$ according to LALD.  
% \end{example}

\begin{figure}[!htbp]
    \centering
    \includegraphics[width=0.48\textwidth]{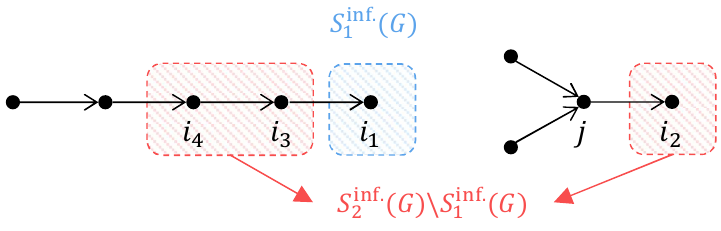}
    \caption{An example of the network, where the marked agents have the relationship as $i_1\succ i_2\succ i_3\succ i_4\succ j$. The 1-influential set and the 2-influential set are represented by dashed borders.}
    \label{fig:lald}
\end{figure}

\begin{theorem}
LALD is an IC 2-selection mechanism.
\end{theorem}

\begin{proof}
For any graph $G\in \mathcal{G}$, We consider three different types of agents.
    \begin{enumerate}
        \item For an agent $i\notin S^{\mathsf{inf.}}_2(G)$, by definition, she can never be in the set by hiding her out-edges. Hence, she will always have 0 probability to be chosen.
        \item For an agent $i\in S^{\mathsf{inf.}}_2(G)\setminus S^{\mathsf{inf.}}_1(G)$ when $S^{\mathsf{inf.}}_2(G)\setminus S^{\mathsf{inf.}}_1(G)\neq \emptyset$, there are two cases. (i) If $i$ is the last agent in the 2-influential set, her probability to be chosen is 1. Hence, she has no incentive to manipulate. (ii) If $i$ is not the last agent, no matter how she hides her out-edges, both the last agent and herself are still in the set $S^{\mathsf{inf.}}_2(G)\setminus S^{\mathsf{inf.}}_1(G)$. Hence, her probability to be chosen remains 0.
        \item For an agent $i\in S^{\mathsf{inf.}}_1(G)$, suppose $S^{\mathsf{inf.}}_1(G) = \{ i_1,\dots, i_q \}$ with $i_t\succ i_{t+1}$ for all $1\leq t< q$. There are three cases. (i) If $i$ is the last agent $i_q$ in the 1-influential set, she has no incentive to manipulate when $S^{\mathsf{inf.}}_2(G)\setminus S^{\mathsf{inf.}}_1(G) = \emptyset$ since $x_i = 1$. When $S^{\mathsf{inf.}}_2(G)\setminus S^{\mathsf{inf.}}_1(G) \neq \emptyset$, no matter how $i$ hides her out-edges, agents in the set $S^{\mathsf{inf.}}_2(G)\setminus S^{\mathsf{inf.}}_1(G)$ will still be in the 2-influential set (even may be in the 1-influential set). After $i$ hides some out-edges, if $S^{\mathsf{inf.}}_2(G')\setminus S^{\mathsf{inf.}}_1(G')\neq \emptyset$, then $i$ cannot have higher probability by Theorem~\ref{thm:lm-ic}; if $S^{\mathsf{inf.}}_2(G')\setminus S^{\mathsf{inf.}}_1(G')= \emptyset$, then $i$ can have at most $1/(1+\ln 2)$ probability since she will no longer be the last agent in $S^{\mathsf{inf.}}_1(G')$, which equals to her original probability. (ii) If $i = i_{q-1}$, she has no incentive to manipulate when $S^{\mathsf{inf.}}_2(G)\setminus S^{\mathsf{inf.}}_1(G) = \emptyset$. It is because $i_q$ will always be in the 1-influential set no matter how $i$ hides her out-edges, which makes increasing her probability to be chosen impossible. When $S^{\mathsf{inf.}}_2(G)\setminus S^{\mathsf{inf.}}_1(G) \neq \emptyset$, it is almost the same condition as that for $i= i_q$. Hence, she cannot increase her probability by manipulation, either. (iii) If $i = i_t$ with $t<q-1$, no matter how she hides her out-edges, $i_q$ and $i_{q-1}$ are always in the 1-influential set. Hence, $i$'s probability to be chosen will not change.
    \end{enumerate}

    Taking all the above together, we can conclude that the mechanism is IC.
\end{proof}

\begin{theorem}
LALD is $\frac{3+\ln2}{4(1+\ln2)}$-optimal.
\end{theorem}

\begin{proof}
Suppose $S^{\mathsf{inf.}}_2(G) = \{ i_1, i_2, \cdots, i_m \}$ for any graph $G\in\mathcal{G}$,. There are two different cases that need consideration.
\begin{enumerate}
    \item If $S^{\mathsf{inf.}}_2(G)\setminus S^{\mathsf{inf.}}_1(G) = \emptyset$, then according to Observation~\ref{ob:1-like} and Theorem~\ref{thm:LM-opt}, we have
    \begin{align*}
    \frac{\mathbb{E}_{S}[\sum_{i\in S} P(i)]}{\sum_{i\in S_2^*} P(i)}&\geq\frac{P(i_m)+\frac{1}{1+\ln2}P(i_1^*)}{P(i_1^*)+P(i_2^*)}\\
    &\geq\frac{\left(\frac{1}{2}+\frac{1}{1+\ln2}\right)P(i_1)}{P(i_1)+P(i_2)}\\
    &\geq\frac{\frac{1}{2}+\frac{1}{1+\ln2}}{2}=\frac{3+\ln2}{4(1+\ln2)}.
    \end{align*}
    \item If $S^{\mathsf{inf.}}_2(G)\setminus S^{\mathsf{inf.}}_1(G) \neq \emptyset$, we first consider the progeny $P(i_m)$. When $i_2^*\in P(i_1^*)$, the structure of the 2-influential set is characterized by Observation~\ref{ob:2-chain}. Then after deleting $i_m$'s out-edges, $i_1\succ i_2$ holds. Hence, $P(i_m)\geq P(i_2) - P(i_m)$ which implies $2P(i_m)\geq P(i_2)$. When $i_2^*\notin P(i_1^*)$, the structure of the 2-influential set is characterized by Observation~\ref{ob:2-single}. Hence, at least one of $2P(i_m)\geq P(i_1)$ and $2P(i_m)\geq P(i_2)$ is satisfied, which all imply that $2P(i_m)\geq P(i_2)$. Therefore, let $P(i_2)/P(i_1) = \rho$ and we have
    \begin{align*}
    \frac{\mathbb{E}_{S}[\sum_{i\in S} P(i)]}{\sum_{i\in S_2^*} P(i)}& \geq\frac{P(i_m)+\frac{1}{1+\ln2}P(i_1^*)}{P(i_1^*)+P(i_2^*)} \\
    & \geq \frac{\frac{1}{2}P(i_2)+\frac{1}{1+\ln2}P(i_1)}{P(i_1)+P(i_2)}\\
    & = \frac{\rho/2 + 1/(1+\ln 2)}{1 + \rho} \\
    & = \frac{\frac{1}{2}(1+\rho) + \frac{1-\ln2}{2(1+\ln2)}}{1+\rho} \\
    & = \frac{1}{2}+\frac{1-\ln2}{2(1+\ln2)}\cdot\frac{1}{1+\rho} \\
    & \geq \frac{1}{2}+\frac{1-\ln2}{4(1+\ln2)}=\frac{3+\ln2}{4(1+\ln2)}.
    \end{align*}
    % \item For the case that $S^{\mathsf{inf.}}_2(G)\setminus S^{\mathsf{inf.}}_1(G) \neq \emptyset$ and $i_2^*\notin P(i_1^*)$, the set $S^{\mathsf{inf.}}_2(G)$ forms a chain $\{i_1,i_2,\cdots i_m\}$ and a pseudo-root $i_q$ by Observation~\ref{ob:2-single}. Then we have
    % \begin{align*}
    % \frac{\mathbb{E}_{S}[\sum_{i\in S} P(i)]}{\sum_{i\in S_2^*} P(i)}&\geq\frac{P(i_m)+\frac{1}{1+\ln2}P(i_1^*)}{P(i_1^*)+P(i_2^*)}\\
    % &\geq\frac{\frac{1}{2}P(i_1)+\frac{1}{1+\ln2}P(i_1)}{P(i_1)+P(i_q)}\\
    % &\geq\frac{\frac{1}{2}+\frac{1}{1+\ln2}}{2}=\frac{3+\ln2}{4(1+\ln2)}.
    % \end{align*}
\end{enumerate}
Therefore, we can conclude that the mechanism is $\frac{3+\ln2}{4(1+\ln2)}$-optimal.
\end{proof}

\section{Discussion}
In this paper, we investigate the incentive-compatible selection mechanisms for one or two influentials, where an agent's influential power is defined by her progeny. The goal is to select agents with progeny as large as possible and to prevent them %the agents 
from hiding their out-edges at the same time. %strategyproof problem of selecting influentials with larger progeny while guaranteeing all agents act truthfully. The challenge here is that agents intend to hide their out-edges for higher selection probabilities. This manipulation can also reduce the progeny of other agents on the same path.
Based on the idea of assigning possibilities of being selected to those agents who can pretend to be the one with the largest or second largest progeny,
% strategically affect the result,
we first propose the $1/(1+\ln2)$-LM mechanism for selecting one agent, which is optimal among all IC and fair single-agent selections. %This mechanism achieves the upper bound of $1/(1+\ln2)$.
We then propose the LALD mechanism %called LALD
for selecting up to two influentials, which has an approximation ratio of $(3+\ln2)/(4(1+\ln2))$ ($\approx 0.54$). To the best of our knowledge, this is the first work to select more than one agent for progeny maximization. There are several interesting future directions worth investigating, and we provide some brief discussions here.

One direction is to narrow the gap between the current lower bound (given by our proposed mechanism) and the current upper bound ($23/27$ as we proved) of the approximation ratio for an optimal IC 2-selection mechanism. For the side of upper bounds, notice that our provided upper bound does not require additional properties like fairness defined in~\cite{zhang2021incentive}. This is because the fairness for selecting a single agent does not apply in selecting multiple agents (e.g., in LALD, the probability of choosing $i_1^*$ may be also related to the structure of the 2-influential set). If we extend the definition of fairness to $k$-fairness like
\begin{definition}[sketch]
$i_1^*$ (or also with $i_2^*$ to $i_k^*$) has the same probability to be chosen when the $k$-influential set and the structure formed by $P(i_1^*)$ (or also with $P(i_2^*)$ to $P(i_k^*)$) remain the same.
\end{definition}
\noindent Then, we can observe that $k$-fairness will become weaker when $k$ becomes larger. \citeauthor{zhang2021incentive}~\shortcite{zhang2021incentive} conjectured that dropping (1-)fairness will not affect the upper bound they characterized. If it can be proven to be true, then we can also draw a corollary that introducing $k$-fairness will not affect the upper bounds of approximation ratios for IC $k$-selection mechanisms. For the side of improving lower bounds, one may consider to utilize more agents in the 2-influential set but not in the 1-influential set. The main difficulty here is these agents %in the 2-influential set but not in the 1-influential set
may have too small progeny when $P(i_2^*) \ll P(i_1^*)$.
% fair(?)

% Impossibilities of extended in 1-set (if there is enough space)

The other direction is to extend the mechanisms for selecting more agents ($k\geq 3$). Similar to the case of selecting two agents, only selecting agents in the $k'$-influential set with $k'<k$ may limit the performance. A natural idea is to select $k$ agents in the $k$-influential set. The main difficulty here is that the structure of the $k$-influential set will become more and more complex when $k$ becomes larger. Intuitively, the structure of the $k$-influential set depends on the relationships among agents $i_1^*$, $\dots$, $i_k^*$. The number of different cases of the structure will grow exponentially with $k$, which is roughly $2^{O(n^2)}$. A possible way to handle this challenge may be recursively considering the influential set with lower $k$.

Finally, in terms of other applications, such as recruiting agents to promote some advertisements, designing selection mechanisms to maximize the expected cardinality of the union of progeny is also a promising future direction.

\section*{Acknowledgements}
This work is supported by Science and Technology Commission of Shanghai Municipality (No. 22ZR1442200 and No. 23010503000), and Shanghai Frontiers Science Center of Human-centered Artificial Intelligence (ShangHAI).

\bibliographystyle{named}
\bibliography{ijcai23}

\begin{thebibliography}{}

\bibitem[\protect\citeauthoryear{Alon \bgroup \em et al.\egroup
  }{2011}]{alon2011sum}
Noga Alon, Felix Fischer, Ariel Procaccia, and Moshe Tennenholtz.
\newblock Sum of us: Strategyproof selection from the selectors.
\newblock In {\em Proceedings of the 13th Conference on Theoretical Aspects of
  Rationality and Knowledge}, pages 101--110, 2011.

\bibitem[\protect\citeauthoryear{Babichenko \bgroup \em et al.\egroup
  }{2018}]{babichenko2018incentive}
Yakov Babichenko, Oren Dean, and Moshe Tennenholtz.
\newblock Incentive-compatible diffusion.
\newblock In {\em Proceedings of the 2018 World Wide Web Conference}, pages
  1379--1388, 2018.

\bibitem[\protect\citeauthoryear{Babichenko \bgroup \em et al.\egroup
  }{2020a}]{babichenko2020incentive1}
Yakov Babichenko, Oren Dean, and Moshe Tennenholtz.
\newblock Incentive-compatible classification.
\newblock In {\em Proceedings of the AAAI Conference on Artificial
  Intelligence}, volume~34, pages 7055--7062, 2020.

\bibitem[\protect\citeauthoryear{Babichenko \bgroup \em et al.\egroup
  }{2020b}]{babichenko2020incentive2}
Yakov Babichenko, Oren Dean, and Moshe Tennenholtz.
\newblock Incentive-compatible selection mechanisms for forests.
\newblock In {\em Proceedings of the 21st ACM Conference on Economics and
  Computation}, pages 111--131, 2020.

\bibitem[\protect\citeauthoryear{Bjelde \bgroup \em et al.\egroup
  }{2017}]{bjelde2017impartial}
Antje Bjelde, Felix Fischer, and Max Klimm.
\newblock Impartial selection and the power of up to two choices.
\newblock {\em ACM Transactions on Economics and Computation (TEAC)},
  5(4):1--20, 2017.

\bibitem[\protect\citeauthoryear{Bousquet \bgroup \em et al.\egroup
  }{2014}]{bousquet2014near}
Nicolas Bousquet, Sergey Norin, and Adrian Vetta.
\newblock A near-optimal mechanism for impartial selection.
\newblock In {\em International Conference on Web and Internet Economics},
  pages 133--146. Springer, 2014.

\bibitem[\protect\citeauthoryear{Caragiannis \bgroup \em et al.\egroup
  }{2021}]{caragiannis2021impartial}
Ioannis Caragiannis, George Christodoulou, and Nicos Protopapas.
\newblock Impartial selection with prior information.
\newblock {\em arXiv preprint arXiv:2102.09002}, 2021.

\bibitem[\protect\citeauthoryear{Caragiannis \bgroup \em et al.\egroup
  }{2022}]{caragiannis2022impartial}
Ioannis Caragiannis, George Christodoulou, and Nicos Protopapas.
\newblock Impartial selection with additive approximation guarantees.
\newblock {\em Theory of Computing Systems}, pages 1--22, 2022.

\bibitem[\protect\citeauthoryear{Cembrano \bgroup \em et al.\egroup
  }{2022}]{DBLP:conf/sigecom/CembranoFHK22}
Javier Cembrano, Felix~A. Fischer, David Hannon, and Max Klimm.
\newblock Impartial selection with additive guarantees via iterated deletion.
\newblock In {\em {EC} '22: The 23rd {ACM} Conference on Economics and
  Computation, Boulder, CO, USA, July 11 - 15, 2022}, pages 1104--1105. {ACM},
  2022.

\bibitem[\protect\citeauthoryear{Fischer and Klimm}{2014}]{fischer2014optimal}
Felix Fischer and Max Klimm.
\newblock Optimal impartial selection.
\newblock In {\em Proceedings of the fifteenth ACM conference on Economics and
  computation}, pages 803--820, 2014.

\bibitem[\protect\citeauthoryear{Ghalme \bgroup \em et al.\egroup
  }{2018}]{ghalme2018design}
Ganesh Ghalme, Sujit Gujar, Amleshwar Kumar, Shweta Jain, and Y~Narahari.
\newblock Design of coalition resistant credit score functions for online
  discussion forums.
\newblock In {\em Proceedings of the 17th International Conference on
  Autonomous Agents and MultiAgent Systems}, pages 95--103, 2018.

\bibitem[\protect\citeauthoryear{Holzman and
  Moulin}{2013}]{holzman2013impartial}
Ron Holzman and Herv{\'e} Moulin.
\newblock Impartial nominations for a prize.
\newblock {\em Econometrica}, 81(1):173--196, 2013.

\bibitem[\protect\citeauthoryear{Kahng \bgroup \em et al.\egroup
  }{2018}]{kahng2018ranking}
Anson Kahng, Yasmine Kotturi, Chinmay Kulkarni, David Kurokawa, and Ariel
  Procaccia.
\newblock Ranking wily people who rank each other.
\newblock In {\em Proceedings of the AAAI Conference on Artificial
  Intelligence}, volume~32, 2018.

\bibitem[\protect\citeauthoryear{Kotturi \bgroup \em et al.\egroup
  }{2020}]{kotturi2020hirepeer}
Yasmine Kotturi, Anson Kahng, Ariel Procaccia, and Chinmay Kulkarni.
\newblock Hirepeer: Impartial peer-assessed hiring at scale in expert
  crowdsourcing markets.
\newblock In {\em Proceedings of the AAAI Conference on Artificial
  Intelligence}, pages 2577--2584, 2020.

\bibitem[\protect\citeauthoryear{Kurokawa \bgroup \em et al.\egroup
  }{2015}]{kurokawa2015impartial}
David Kurokawa, Omer Lev, Jamie Morgenstern, and Ariel~D Procaccia.
\newblock Impartial peer review.
\newblock In {\em Twenty-Fourth International Joint Conference on Artificial
  Intelligence}, 2015.

\bibitem[\protect\citeauthoryear{Mackenzie}{2015}]{mackenzie2015symmetry}
Andrew Mackenzie.
\newblock Symmetry and impartial lotteries.
\newblock {\em Games and Economic Behavior}, 94:15--28, 2015.

\bibitem[\protect\citeauthoryear{Mattei \bgroup \em et al.\egroup
  }{2020}]{DBLP:conf/ijcai/MatteiTZ20}
Nicholas Mattei, Paolo Turrini, and Stanislav Zhydkov.
\newblock Peernomination: Relaxing exactness for increased accuracy in peer
  selection.
\newblock In {\em Proceedings of the Twenty-Ninth International Joint
  Conference on Artificial Intelligence, {IJCAI} 2020}, pages 393--399.
  ijcai.org, 2020.

\bibitem[\protect\citeauthoryear{Olckers and
  Walsh}{2022}]{olckers2022manipulation}
Matthew Olckers and Toby Walsh.
\newblock Manipulation and peer mechanisms: A survey.
\newblock {\em arXiv preprint arXiv:2210.01984}, 2022.

\bibitem[\protect\citeauthoryear{Page \bgroup \em et al.\egroup
  }{1999}]{page1999pagerank}
Lawrence Page, Sergey Brin, Rajeev Motwani, and Terry Winograd.
\newblock The pagerank citation ranking: Bringing order to the web.
\newblock Technical report, Stanford InfoLab, 1999.

\bibitem[\protect\citeauthoryear{Wang \bgroup \em et al.\egroup
  }{2018}]{wang2018tsp}
Yufeng Wang, Hui Fang, Chonghu Cheng, and Qun Jin.
\newblock Tsp: Truthful grading-based strategyproof peer selection for moocs.
\newblock In {\em 2018 IEEE International Conference on Teaching, Assessment,
  and Learning for Engineering (TALE)}, pages 679--684. IEEE, 2018.

\bibitem[\protect\citeauthoryear{Zhang \bgroup \em et al.\egroup
  }{2021}]{zhang2021incentive}
Xiuzhen Zhang, Yao Zhang, and Dengji Zhao.
\newblock Incentive compatible mechanism for influential agent selection.
\newblock In {\em International Symposium on Algorithmic Game Theory}, pages
  79--93. Springer, 2021.

\end{thebibliography}

\end{document}